\documentclass[11pt]{article}
\usepackage{hyperref}
\usepackage{times}  
\usepackage{mathpazo}
\usepackage{amssymb,amsmath,amsthm}
\usepackage{epsfig}
\usepackage{amsmath}
\usepackage{fullpage}

\newtheorem{theorem}{Theorem}[section]
\newtheorem{proposition}[theorem]{Proposition}
\newtheorem{definition}[theorem]{Definition}

\newtheorem{claim}[theorem]{Claim}
\newtheorem{lemma}[theorem]{Lemma}

\newtheorem{corollary}[theorem]{Corollary}

\newtheorem{remark}[theorem]{Remark}

\newcommand{\qedsymb}{\hfill{\rule{2mm}{2mm}}}
\renewenvironment{proof}[1][]{\begin{trivlist}
\item[\hspace{\labelsep}{\bf\noindent Proof#1:\/}] }{\qedsymb\end{trivlist}}

\def\calS{{\cal S}}
\def\calO{{\cal O}}

\def\R{\mathbb{R}}

\def\N{\mathbb{N}}

\newcommand{\IC}{\mathsf{Index\mbox{-}Coding}}
\newcommand{\NP}{\mathsf{NP}}

\renewcommand{\P}{\mathsf{P}}

\newcommand{\YES}{\mathsf{YES}}
\newcommand{\NO}{\mathsf{NO}}

\newcommand{\od}{\overline{\xi}}

\newcommand{\eps}{\epsilon}
\renewcommand{\epsilon}{\varepsilon}

\newcommand{\rank}{\mathop{\mathrm{rank}}}
\newcommand{\minrank}{\mathop{\mathrm{minrk}}}

\newcommand{\linspan}{\mathop{\mathrm{span}}}
\newcommand{\Fset}{\mathbb{F}}         

\begin{document}

\title{{\bf Improved NP-Hardness of Approximation for Orthogonality Dimension and Minrank}}
\author{
Dror Chawin\thanks{School of Computer Science, The Academic College of Tel Aviv-Yaffo, Tel Aviv 61083, Israel. Research supported by the Israel Science Foundation (grant No.~1218/20).}
\and
Ishay Haviv\footnotemark[1]
}

\date{}

\maketitle

\begin{abstract}
The orthogonality dimension of a graph $G$ over $\R$ is the smallest integer $k$ for which one can assign a nonzero $k$-dimensional real vector to each vertex of $G$, such that every two adjacent vertices receive orthogonal vectors. We prove that for every sufficiently large integer $k$, it is $\NP$-hard to decide whether the orthogonality dimension of a given graph over $\R$ is at most $k$ or at least $2^{(1-o(1)) \cdot k/2}$. We further prove such hardness results for the orthogonality dimension over finite fields as well as for the closely related minrank parameter, which is motivated by the index coding problem in information theory.
This in particular implies that it is $\NP$-hard to approximate these graph quantities to within any constant factor.
Previously, the hardness of approximation was known to hold either assuming certain variants of the Unique Games Conjecture or for approximation factors smaller than $3/2$.
The proofs involve the concept of line digraphs and bounds on their orthogonality dimension and on the minrank of their complement.
\end{abstract}

\section{Introduction}

A graph $G$ is said to be $k$-colorable if its vertices can be colored by $k$ colors such that every two adjacent vertices receive distinct colors.
The chromatic number of $G$, denoted by $\chi(G)$, is the smallest integer $k$ for which $G$ is $k$-colorable.
As a fundamental and popular graph quantity, the chromatic number has received a considerable amount of attention in the literature from a computational perspective, as described below.

The problem of deciding whether a graph $G$ satisfies $\chi(G) \leq 3$ is one of the classical twenty-one $\NP$-complete problems presented by Karp~\cite{Karp72} in 1972.
Khanna, Linial, and Safra~\cite{KhannaLS00} proved that it is $\NP$-hard to distinguish between graphs $G$ that satisfy $\chi(G) \leq 3$ from those satisfying $\chi(G) \geq 5$. This result, combined with the approach of Garey and Johnson~\cite{GareyJ76a} and with a result of Stahl~\cite{Stahl76}, implies that for every $k \geq 6$, it is $\NP$-hard to decide whether a graph $G$ satisfies $\chi(G) \leq k$ or $\chi(G) \geq 2k-2$.
Brakensiek and Guruswami~\cite{BrakensiekG16} proved that for every $k \geq 3$, it is $\NP$-hard to distinguish between the cases $\chi(G) \leq k$ and $\chi(G) \geq 2k-1$, and the $2k-1$ bound was further improved to $2k$ by Barto, Bul{\'{\i}}n, Krokhin, and Opr{\v s}al~\cite{BartoBKO21}.
For large values of $k$, it was shown by Khot~\cite{Khot01} that it is $\NP$-hard to decide whether a graph $G$ satisfies $\chi(G) \leq k$ or $\chi(G) \geq k^{\Omega(\log k)}$, and the latter condition was strengthened to $\chi(G) \geq 2^{k^{1/3}}$ by Huang~\cite{Huang13}.
A substantial improvement was recently obtained by Wrochna and {\v Z}ivn{\' y}~\cite{WZ20}, who proved that for every $k \geq 4$, it is $\NP$-hard to decide whether a given graph $G$ satisfies $\chi(G) \leq k$ or $\chi(G) \geq \binom{k}{\lfloor k/2 \rfloor }$.
The proof of this result combined the hardness result of~\cite{Huang13} with the construction of line digraphs~\cite{HararyN60} and with a result of Poljak and R{\"{o}}dl~\cite{PoljakR81}.
Note that under certain variants of the Unique Games Conjecture, stronger hardness results are known to hold, namely, hardness of deciding whether a given graph $G$ satisfies $\chi(G) \leq k_1$ or $\chi(G) \geq k_2$ for all integers $k_2 > k_1 \geq 3$~\cite{DinurMR06} (see also~\cite{DinurS10}).

The present paper studies the computational complexity of algebraic variants of the chromatic number of graphs.
A $k$-dimensional orthogonal representation of a graph $G=(V,E)$ over a field $\Fset$ is an assignment of a vector $u_v \in \Fset^k$ with $\langle u_v, u_v \rangle \neq 0$ to each vertex $v \in V$, such that for every two adjacent vertices $v$ and $v'$ it holds that $\langle u_v, u_{v'} \rangle =0$. Here, for two vectors $x,y \in \Fset^k$, we consider the standard inner product defined by $\langle x,y \rangle = \sum_{i=1}^{k}{x_i y_i}$ with operations over $\Fset$.
The orthogonality dimension of $G$ over $\Fset$, denoted by $\od_\Fset(G)$, is the smallest integer $k$ for which $G$ admits a $k$-dimensional orthogonal representation over $\Fset$ (see Remark~\ref{remark:or-comp}).
It can be easily seen that for every graph $G$ and for every field $\Fset$, it holds that $\od_\Fset(G) \leq \chi(G)$. In addition, if $\Fset$ is a fixed finite field or the real field $\R$, it further holds that $\od_\Fset(G) \geq \Omega (\log \chi(G))$.
Both bounds are known to be tight in the worst case (see Claim~\ref{claim:mr_od_chi} and~\cite[Chapter~10]{LovaszBook}).
The study of orthogonal representations and orthogonality dimension was initiated in the seminal work of Lov{\'{a}}sz~\cite{Lovasz79} on the $\vartheta$-function and has found applications in various areas, e.g., information theory~\cite{Lovasz79}, graph theory~\cite{LovaszSS89}, and quantum communication complexity~\cite[Chapter~8.5]{deWolfThesis}.

The interest in the hardness of determining the orthogonality dimension of graphs dates back to a paper of Lov{\'{a}}sz, Saks, and Schrijver~\cite{LovaszSS89}, where it was noted that the problem seems difficult.
The aforementioned relations between the chromatic number and the orthogonality dimension yield that hardness of deciding whether a graph $G$ satisfies $\chi(G) \leq k_1$ or $\chi(G) \geq k_2$ implies the hardness of deciding whether it satisfies $\od_\Fset(G) \leq k_1$ or $\od_\Fset(G) \geq \Omega(\log k_2)$, provided that $\Fset$ is a finite field or $\R$.
It therefore follows from~\cite{DinurMR06} that assuming certain variants of the Unique Games Conjecture, it is hard to decide whether a graph $G$ satisfies $\od_\Fset(G) \leq k_1$ or $\od_\Fset(G) \geq k_2$ for all integers $k_2 > k_1 \geq 3$.
This reasoning, however, does not yield $\NP$-hardness results for the orthogonality dimension (without additional complexity assumptions), even using the strongest known $\NP$-hardness results of the chromatic number.
Yet, a result of Peeters~\cite{Peeters96} implies that for every field $\Fset$, it is $\NP$-hard to decide if a given graph $G$ satisfies $\od_\Fset(G) \leq 3$, hence it is $\NP$-hard to approximate the orthogonality dimension of a graph over $\Fset$ to within any factor smaller than $4/3$. Over the reals, the hardness of approximation for the orthogonality dimension was recently extended in~\cite{GolovnevH20} to any factor smaller than $3/2$.

Another algebraic quantity of graphs is the minrank parameter that was introduced in 1981 by Haemers~\cite{Haemers81} in the study of the Shannon capacity of graphs.
The minrank parameter was used in~\cite{Haemers79,Haemers81} to answer questions of Lov{\'{a}}sz~\cite{Lovasz79} and was later applied by Alon~\cite{AlonUnion98}, with a different formulation, to disprove a conjecture of Shannon~\cite{Shannon56}.
The minrank of a graph $G$ over a field $\Fset$, denoted by ${\minrank}_\Fset(G)$, is closely related to the orthogonality dimension of the complement graph $\overline{G}$ over $\Fset$ and satisfies ${\minrank}_\Fset(G) \leq \od_\Fset(\overline{G})$.
The difference between the two quantities comes, roughly speaking, from the fact that the definition of minrank involves the notion of orthogonal bi-representations rather than orthogonal representations (for the precise definitions, see Section~\ref{sec:mr_od}).
The study of the minrank parameter is motivated by various applications in information theory and in theoretical computer science. A prominent one is the well-studied index coding problem, for which the minrank parameter perfectly characterizes the optimal length of its linear solutions, as was shown by Bar-Yossef, Birk, Jayram, and Kol~\cite{BBJK06} (see Section~\ref{sec:pre_index}).

Similarly to the situation of the orthogonality dimension, it was proved in~\cite{Peeters96} that for every field $\Fset$, it is $\NP$-hard to decide if a given graph $G$ satisfies ${\minrank}_\Fset(G) \leq 3$. It was further shown by Dau, Skachek, and Chee~\cite{DauSC14} that it is $\NP$-hard to decide whether a given digraph $G$ satisfies ${\minrank}_{\Fset_2}(G) \leq 2$. Note that for (undirected) graphs, the minrank over any field is at most $2$ if and only if the complement graph is bipartite, a property that can be checked in polynomial time.
Motivated by the computational aspects of the index coding problem, Langberg and Sprintson~\cite{LangbergS08} related the minrank of a graph to the chromatic number of its complement and derived from~\cite{DinurMR06} that assuming certain variants of the Unique Games Conjecture, it is hard to decide whether a given graph $G$ satisfies ${\minrank}_\Fset(G) \leq k_1$ or ${\minrank}_\Fset(G) \geq k_2$, provided that $k_2 > k_1 \geq 3$ and that $\Fset$ is a finite field.
Similar hardness results were obtained in~\cite{LangbergS08} for additional settings of the index coding problem, including the general (non-linear) index coding problem over a constant-size alphabet.

\subsection{Our Contribution}

This paper provides improved $\NP$-hardness of approximation results for the orthogonality dimension and for the minrank parameter over various fields.
We start with the following result, which is concerned with the orthogonality dimension over the reals.

\begin{theorem}\label{thm:IntroR}
There exists a function $f:\N \rightarrow \N$ satisfying $f(k) = 2^{(1-o(1)) \cdot k/2}$ such that for every sufficiently large integer $k$, it is $\NP$-hard to decide whether a given graph $G$ satisfies
\[\od_\R(G) \leq k~\mbox{~~~ or ~~~}~\od_\R(G) \geq f(k).\]
\end{theorem}
\noindent
Theorem~\ref{thm:IntroR} implies that it is $\NP$-hard to approximate the orthogonality dimension of a graph over the reals to within any constant factor.
Previously, such an $\NP$-hardness result was known to hold only for approximation factors smaller than $3/2$~\cite{GolovnevH20}.

We proceed with the following result, which is concerned with the orthogonality dimension and the minrank parameter over finite fields.
\begin{theorem}\label{thm:IntroF}
For every finite field $\Fset$, there exists a function $f:\N \rightarrow \N$ satisfying $f(k) = 2^{(1-o(1)) \cdot k/2}$ such that for every sufficiently large integer $k$, the following holds.
\begin{enumerate}
  \item It is $\NP$-hard to decide whether a given graph $G$ satisfies $\od_\Fset(G) \leq k$ or $\od_\Fset(G) \geq f(k)$.
  \item It is $\NP$-hard to decide whether a given graph $G$ satisfies ${\minrank}_\Fset(G) \leq k$ or ${\minrank}_\Fset(G) \geq f(k)$.
\end{enumerate}
\end{theorem}
\noindent
Theorem~\ref{thm:IntroF} implies that over any finite field, it is $\NP$-hard to approximate the orthogonality dimension and the minrank of a graph to within any constant factor.
Let us stress that this hardness result relies solely on the assumption $\P \neq \NP$ rather than on stronger complexity assumptions and thus settles a question raised in~\cite{LangbergS08}.
Prior to this work, it was known that it is $\NP$-hard to approximate the minrank of graphs to within any factor smaller than $4/3$~\cite{Peeters96} and the minrank of digraphs over $\Fset_2$ to within any factor smaller than $3/2$~\cite{DauSC14}.

A central component of the proofs of Theorems~\ref{thm:IntroR} and~\ref{thm:IntroF} is the notion of line digraphs, introduced in~\cite{HararyN60}, that was first used in the context of hardness of approximation by Wrochna and {\v Z}ivn{\' y}~\cite{WZ20} (see also~\cite{GuruswamiS20}).
It was shown in~\cite{HarnerE72,PoljakR81} that the chromatic number of any graph is exponential in the chromatic number of its line digraph.
This result was iteratively applied by the authors of~\cite{WZ20} to improve the $\NP$-hardness of the chromatic number from the $k$ vs.~$2^{k^{1/3}}$ gap of~\cite{Huang13} to their $k$ vs.~$\binom{k}{\lfloor k/2 \rfloor}$ gap.
The main technical contribution of the present work lies in analyzing the orthogonality dimension of line digraphs and the minrank parameter of their complement.
We actually show that on line digraphs, these graph parameters are quadratically related to the chromatic number.
This allows us to derive our hardness results from the hardness of the chromatic number given in~\cite{WZ20}, where the obtained gaps are only quadratically weaker.
We further discuss some limitations of our approach, involving an analogue of Sperner's theorem for subspaces due to Kalai~\cite{Kalai80}.

To demonstrate our combinatorial contribution, consider a finite field $\Fset$ and a graph $G$, and let $H$ denote the underlying graph of the line digraph associated with $G$ (see Definition~\ref{def:line}).
It is shown in~\cite{HarnerE72,PoljakR81} that $\chi(H)$ is the smallest integer $n$ such that $\chi(G) \leq \binom{n}{\lfloor n/2 \rfloor}$. In particular, if $\chi(G) \leq \binom{n}{\lfloor n/2 \rfloor}$ then $ \od_\Fset(H) \leq \chi(H) \leq n$. On the other hand, we prove that if $\od_\Fset(H) \leq n$ then $\chi(G) \leq |\Fset|^{n^2}$, which implies that \[\od_\Fset(H) \geq \sqrt{\log_{|\Fset|}{\chi(G)}}.\]
By combining these bounds, it follows that $\od_\Fset(H)$ and $\chi(H)$ are quadratically related.
Results of this nature are also proved for the orthogonality dimension over the reals and for the minrank parameter over finite fields (see Theorems~\ref{thm:od(H)},~\ref{thm:od_R(H)}, and~\ref{thm:mr(H)}).

We finally show that our approach might be useful for proving hardness results for the general (non-linear) index coding problem over a constant-size alphabet, for which no $\NP$-hardness result is currently known.
It was shown by Langberg and Sprintson~\cite{LangbergS08} that for an instance of the index coding problem represented by a graph $G$, the length of an optimal solution is at most $\chi(\overline{G})$ and at least $\Omega(\log \log \chi(\overline{G}))$.
It thus follows that an $\NP$-hardness result for the chromatic number with a double-exponential gap would imply an $\NP$-hardness result for the general index coding problem.
However, no such $\NP$-hardness result is currently known for the chromatic number without relying on further complexity assumptions.
To tackle this issue, we study the index coding problem on instances which are complement of line digraphs (see Theorem~\ref{thm:index_b}). As a consequence of our results, we obtain that the $\NP$-hardness of the general index coding problem can be derived from an $\NP$-hardness result of the chromatic number with only a single-exponential gap, not that far from the best known gap given in~\cite{WZ20}. For a precise statement, see Theorem~\ref{thm:index}.

\subsection{Related Work}
We gather here several related results from the literature.

\begin{itemize}
  \item A result of Zuckerman~\cite{Zuckerman07} asserts that for any $\eps >0 $, it is $\NP$-hard to approximate the chromatic number of a graph on $n$ vertices to within a factor of $n^{1-\eps}$.
It would be interesting to figure out if such a hardness result holds for the orthogonality dimension and for the minrank parameter.
The present paper, however, focuses on the hardness of gap problems with constant thresholds, independent of the number of vertices.
  \item As mentioned earlier, Peeters~\cite{Peeters96} proved that for every field $\Fset$, it is $\NP$-hard to decide if the minrank (or the orthogonality dimension) of a given graph is at most $3$.
        We note that for finite fields, this can also be derived from a result of Hell and Ne\v{s}et\v{r}il~\cite{HellN90}.
  \item For the chromatic number of hypergraphs, the gaps for which $\NP$-hardness is known to hold are much stronger than for graphs.
For example, it was shown in~\cite{Bhangale18} that for some $\delta >0$, it is $\NP$-hard to decide if a given $4$-uniform hypergraph $G$ on $n$ vertices satisfies $\chi(G) \leq 2$ or $\chi(G) \geq \log^\delta n$. An analogous result for the orthogonality dimension of hypergraphs over $\R$ was proved in~\cite{HavivMFCS19}.
  \item On the algorithmic side, a long line of work has explored the number of colors that an efficient algorithm needs for properly coloring a given $k$-colorable graph, where $k \geq 3$ is a fixed constant.
For example, there exists a polynomial-time algorithm that on a given $3$-colorable graph with $n$ vertices uses $O(n^{0.19996})$ colors~\cite{KT17}.
Algorithms of this nature exist for the graph parameters studied in this work as well. Indeed, there exists a polynomial-time algorithm that given a graph $G$ on $n$ vertices with $\od_\R(G) \leq 3$ finds a proper coloring of $G$ with $O(n^{0.2413})$ colors~\cite{HavivMFCS19}. Further, there exists a polynomial-time algorithm that given a graph $G$ on $n$ vertices with ${\minrank}_{\Fset_2}(\overline{G}) \leq 3$ finds a proper coloring of $G$ with $O(n^{0.2574})$ colors~\cite{ChlamtacH14}.
\end{itemize}

\subsection{Outline}
The rest of the paper is organized as follows.
In Section~\ref{sec:preliminaries}, we collect several definitions and results that will be used throughout this paper.
In Section~\ref{sec:line}, we study the underlying graphs of line digraphs and their behavior with respect to the orthogonality dimension, the minrank parameter, and the index coding problem. We also discuss there some limitations of our approach, given in Sections~\ref{sec:clique_S1} and~\ref{sec:S2_R}.
Finally, in Section~\ref{sec:hard}, we prove our hardness results and complete the proofs of Theorems~\ref{thm:IntroR} and~\ref{thm:IntroF}.

\section{Preliminaries}\label{sec:preliminaries}

For an integer $n$, we use the notation $[n] = \{1,2,\ldots,n\}$. All the logarithms are in base $2$ unless otherwise specified.
Throughout the paper, undirected graphs are referred to as graphs, and directed graphs are referred to as digraphs.
All the considered graphs and digraphs are simple.
A homomorphism from a graph $G_1=(V_1,E_1)$ to a graph $G_2=(V_2,E_2)$ is a function $g:V_1 \rightarrow V_2$ such that for every two vertices $x,y \in V_1$ with $\{x,y\} \in E_1$, it holds that $\{g(x),g(y)\} \in E_2$.

\subsection{Orthogonality Dimension and Minrank}\label{sec:mr_od}
For a field $\Fset$, an integer $k$, and two vectors $x,y \in \Fset^k$, let $\langle x,y \rangle = \sum_{i=1}^{k}{x_i y_i}$ denote the standard inner product of $x$ and $y$ over $\Fset$.
If $\langle x,y \rangle = 0$ then the vectors $x$ and $y$ are called orthogonal (over $\Fset$). If $\langle x,x \rangle = 0$ then the vector $x$ is called self-orthogonal (and otherwise, non-self-orthogonal).
The orthogonality dimension of a graph is defined as follows (see, e.g.,~\cite[Chapter~11]{LovaszBook}).
\begin{definition}[Orthogonality Dimension]\label{def:od}
A {\em $k$-dimensional orthogonal representation} of a graph $G=(V,E)$ over a field $\Fset$ is an assignment of a vector $u_v \in \mathbb{F}^k$ with $\langle u_v,u_v \rangle \neq 0$ to each vertex $v \in V$, such that $\langle u_v, u_{v'} \rangle = 0$ whenever $v$ and $v'$ are adjacent vertices in $G$.
The {\em orthogonality dimension} of a graph $G$ over a field $\Fset$, denoted by $\od_\Fset(G)$, is the smallest integer $k$ for which there exists a $k$-dimensional orthogonal representation of $G$ over $\Fset$.
\end{definition}

\begin{remark}\label{remark:or-comp}
We note that orthogonal representations are sometimes defined in the literature such that the vectors associated with {\em non-adjacent} vertices are required to be orthogonal, that is, as orthogonal representations of the complement graph. While we find it more convenient to use the other definition in this paper, one can view the notation $\od_\Fset(G)$ as standing for $\xi_\Fset(\overline{G})$, i.e., the orthogonality dimension of the complement graph.
\end{remark}

The orthogonality dimension of graphs can be equivalently expressed in terms of graph homomorphisms.
This requires the following family of graphs.

\begin{definition}\label{def:O(F,n)}
For a field $\Fset$ and an integer $k$, let $\calO(\Fset,k)$ denote the graph whose vertices are all the non-self-orthogonal vectors in $\Fset^k$, where two distinct vectors are adjacent if they are orthogonal over $\Fset$.
\end{definition}

The following characterization of the orthogonality dimension follows directly from Definitions~\ref{def:od} and~\ref{def:O(F,n)}.
\begin{proposition}\label{prop:od_def}
For every field $\Fset$ and for every graph $G$, $\od_\Fset(G)$ is the smallest integer $k$ for which there exists a homomorphism from $G$ to $\calO(\Fset,k)$.
\end{proposition}
\noindent
It follows from Proposition~\ref{prop:od_def} that the orthogonality dimension over a field $\Fset$ is monotone under homomorphisms, namely, if there exists a homomorphism from a graph $G_1$ to a graph $G_2$, then $\od_\Fset(G_1) \leq \od_\Fset(G_2)$.

The minrank parameter, introduced in~\cite{Haemers81}, is defined as follows.

\begin{definition}[Minrank]\label{def:minrank}
Let $G=(V,E)$ be a digraph on the vertex set $V=[n]$, and let $\Fset$ be a field.
We say that a matrix $M \in \Fset^{n \times n}$ {\em represents} $G$ if $M_{i,i} \neq 0$ for every $i \in V$, and $M_{i,j}=0$ for every distinct vertices $i,j \in V$ such that $(i,j) \notin E$.
The {\em minrank} of $G$ over $\Fset$ is defined as
\[{\minrank}_\Fset(G) =  \min\{{\rank}_{\Fset}(M)\mid M \mbox{ represents }G\mbox{ over }\Fset\}.\]
The definition is naturally extended to graphs by replacing every edge with two oppositely directed edges.
\end{definition}

We next describe an alternative definition, due to Peeters~\cite{Peeters96}, for the minrank of graphs.
This requires the following family of graphs.

\begin{definition}\label{def:O'(F,k)}
For a field $\Fset$ and an integer $k$, let $\calO'(\Fset,k)$ denote the graph whose vertices are all the pairs $(u,w) \in \Fset^k \times \Fset^k$ such that $\langle u,w\rangle \neq 0$, where two distinct pairs $(u_1,w_1)$ and $(u_2,w_2)$ are adjacent if they satisfy $\langle u_1,w_2 \rangle = \langle u_2, w_1 \rangle = 0$.
\end{definition}
\begin{proposition}[\cite{Peeters96}]\label{prop:minrk_def}
For every field $\Fset$ and for every graph $G$, ${\minrank}_\Fset(\overline{G})$ is the smallest integer $k$ for which there exists a homomorphism from $G$ to $\calO'(\Fset,k)$.
\end{proposition}
\noindent
It follows from Proposition~\ref{prop:minrk_def} that the minrank of the complement over a field $\Fset$ is monotone under homomorphisms, namely, if there exists a homomorphism from a graph $G_1$ to a graph $G_2$, then ${\minrank}_\Fset(\overline{G_1}) \leq {\minrank}_\Fset(\overline{G_2})$.

The following claim summarizes some known relations between the studied graph parameters.
We provide a quick proof for completeness.
\begin{claim}\label{claim:mr_od_chi}
For every field $\Fset$ and for every graph $G$, it holds that
\[ {\minrank}_\Fset(\overline{G}) \leq \od_\Fset(G) \leq \chi(G).\]
In addition, if $\Fset$ is finite, then
\[ {\minrank}_\Fset(\overline{G}) \geq \log_{|\Fset|}\chi(G).\]
\end{claim}
\begin{proof}
The inequality ${\minrank}_\Fset(\overline{G}) \leq \od_\Fset(G)$ follows by combining Propositions~\ref{prop:od_def} and~\ref{prop:minrk_def} with the fact that for every integer $k$, the graph $\calO(\Fset,k)$ admits a homomorphism to the graph $\calO'(\Fset,k)$, mapping any vertex $u$ to the pair $(u,u)$.

For the inequality $\od_\Fset(G) \leq \chi(G)$, combine Proposition~\ref{prop:od_def} with the fact that for every integer $k$, the complete graph on $k$ vertices admits a homomorphism to the graph $\calO(\Fset,k)$, mapping the $i$th vertex to the $i$th vector of the standard basis of $\Fset^k$.

Next, assuming that $\Fset$ is finite, we show that ${\minrank}_\Fset(\overline{G}) \geq \log_{|\Fset|}\chi(G)$.
By Proposition~\ref{prop:minrk_def}, it suffices to show that for every integer $k$, the graph $\calO'(\Fset,k)$ admits a proper coloring with $|\Fset|^k$ colors.
To see this, assign to every vertex $(u,w)$ of $\calO'(\Fset,k)$ the vector $u$ as a color. Notice that for two adjacent vertices $(u_1,w_1)$ and $(u_2,w_2)$ of $\calO'(\Fset,k)$, it holds that $u_1 \neq u_2$, because $u_1$ is orthogonal to $w_2$ whereas $u_2$ is not. This completes the proof.
\end{proof}

\subsection{Index Coding}\label{sec:pre_index}

The index coding problem, introduced in~\cite{BBJK06}, is concerned with economical strategies for broadcasting information to $n$ receivers in a way that enables each of them to retrieve its own message, a symbol from some given alphabet $\Sigma$. For this purpose, each receiver is allowed to use some prior side information that consists of a subset of the messages required by the other receivers. The side information map is naturally represented by a digraph on $[n]$, which includes an edge $(i,j)$ if the $i$th receiver knows the message required by the $j$th receiver. The objective is to minimize the length of the transmitted information. For simplicity, we consider here the case of symmetric side information maps, represented by graphs rather than by digraphs.
The formal definition follows.

\begin{definition}[Index Coding]\label{def:index}
Let $G$ be a graph on the vertex set $[n]$, and let $\Sigma$ be an alphabet.
An {\em index code} for $G$ over $\Sigma$ of length $k$ is an encoding function $E:\Sigma^n \rightarrow \Sigma^k$ such that for every $i \in [n]$, there exists a decoding function $g_i :\Sigma^{k+|N_G(i)|} \rightarrow  \Sigma$, such that for every $x \in \Sigma^n$, it holds that $g_i(E(x),x|_{N_G(i)}) = x_i$.
Here, $N_G(i)$ stands for the set of vertices in $G$ adjacent to the vertex $i$, and $x|_{N_G(i)}$ stands for the restriction of $x$ to the indices of $N_G(i)$.
If $\Sigma$ is a field $\Fset$ and the encoding function $E$ is linear over $\Fset$, then we say that the index code is {\em linear over $\Fset$}.
\end{definition}

Bar-Yossef et al.~\cite{BBJK06} showed that the minrank parameter characterizes the length of optimal solutions to the index coding problem in the linear setting.
\begin{proposition}[\cite{BBJK06}]\label{prop:ic_minrk}
For every field $\Fset$ and for every graph $G$, the minimal length of a linear index code for $G$ over $\Fset$ is ${\minrank}_\Fset(G)$.
\end{proposition}

\subsection{Computational Problems}
Throughout the paper, we consider computational decision problems associated with several graph quantities: chromatic number, orthogonality dimension over a given field, minrank over a given field, and the minimal length of an index code over a given alphabet. The problems are considered in their {\em promise} version, defined as follows. Let $\psi$ be a graph quantity, and let $k_1 < k_2$ be two integers. We consider the problem of deciding whether a graph $G$ satisfies $\psi(G) \leq k_1$ or $\psi(G) \geq k_2$. In this problem, the input is a graph $G$ that is promised to satisfy either $\psi(G) \leq k_1$ or $\psi(G) \geq k_2$, where in the former case, $G$ is referred to as a $\YES$ instance, and in the latter as a $\NO$ instance. The goal is to distinguish between the two cases. Note that if this problem is $\NP$-hard, then it is $\NP$-hard to approximate the value of $\psi(G)$ for a given graph $G$ to within any factor smaller than $k_2/k_1$.

\section{Line Digraphs}\label{sec:line}

In 1960, Harary and Norman~\cite{HararyN60} introduced the concept of line digraphs, defined as follows.
\begin{definition}[Line Digraph]\label{def:line}
For a digraph $G = (V,E)$, the {\em line digraph of $G$}, denoted by $\delta G$, is the digraph on the vertex set $E$ that includes a directed edge from a vertex $(x,y)$ to a vertex $(z,w)$ whenever $y=z$.
\end{definition}
\noindent
Definition~\ref{def:line} is naturally extended to graphs $G$ by replacing every edge of $G$ with two oppositely directed edges. Note that in this case, the number of vertices in $\delta G$ is twice the number of edges in $G$. We will frequently consider the underlying graph of the digraph $\delta G$, i.e., the graph obtained from $\delta G$ by ignoring the directions of the edges.

The following result of Poljak and R{\"{o}}dl~\cite{PoljakR81}, which strengthens a previous result of Harner and Entringer~\cite{HarnerE72}, shows that the chromatic number of a graph $G$ precisely determines the chromatic number of the underlying graph of $\delta G$.
The statement of the result uses the function $b: \N \rightarrow \N$ defined by $b(n) = \binom{n}{\lfloor n/2 \rfloor}$.

\begin{theorem}[\cite{HarnerE72,PoljakR81}]\label{thm:chi_delta}
Let $G$ be a graph, and let $H$ be the underlying graph of the digraph $\delta G$. Then,
\[\chi(H) = \min \{ n \mid \chi(G) \leq b(n) \}.\]
\end{theorem}
\noindent
Using the fact that $b(n) \sim \frac{2^n}{\sqrt{\pi n /2}}$, Theorem~\ref{thm:chi_delta} implies that the chromatic number of $G$ is exponential in the chromatic number of $H$.
Our goal in this section is to relate the chromatic number of $G$ to other graph parameters of $H$, namely, the orthogonality dimension, the minrank of the complement, and the optimal length of an index code for the complement.

\subsection{Orthogonality Dimension}\label{sec:od_line}

For a field $\Fset$, an integer $n$, and a subspace $U$ of $\Fset^n$, we denote by $U^{\perp}$ the subspace of $\Fset^n$ that consists of the vectors that are orthogonal to $U$ over $\Fset$, i.e.,
\[U^{\perp} = \{ w \in \Fset^n \mid \langle w,u \rangle = 0~~\mbox{for every}~u \in U \}.\]
Consider the following family of graphs.
\begin{definition}\label{def:S1}
For a field $\Fset$ and an integer $n$, let $\calS(\Fset,n)$ denote the graph whose vertices are all the subspaces of $\Fset^n$, where two distinct subspaces $U_1$ and $U_2$ are adjacent if there exists a vector $w \in \Fset^n$ with $\langle w,w \rangle \neq 0$ that satisfies $w \in U_1 \cap U_2^{\perp}$ and, in addition, there exists a vector $w' \in \Fset^n$ with $\langle w',w' \rangle \neq 0$ that satisfies $w' \in U_2 \cap U_1^{\perp}$.
\end{definition}
\noindent
In words, two subspaces of $\Fset^n$ are adjacent in the graph $\calS(\Fset,n)$ if each of them includes a non-self-orthogonal vector that is orthogonal to the entire other subspace.
Note that for an infinite field $\Fset$ and for $n \geq 2$, the vertex set of $\calS(\Fset,n)$ is infinite.

We argue that the chromatic number of a graph $G$ can be used to estimate the orthogonality dimension of the underlying graph $H$ of its line digraph $\delta G$.
First, recall that by Theorem~\ref{thm:chi_delta}, the chromatic number of $H$ is logarithmic in $\chi(G)$.
This implies, using Claim~\ref{claim:mr_od_chi}, that the orthogonality dimension of $H$ over any field is at most logarithmic in $\chi(G)$.
For a lower bound on the orthogonality dimension of $H$, we need the following lemma that involves the graphs $\calS(\Fset,n)$.

\begin{lemma}\label{lemma:od_delta_new}
Let $\Fset$ be a field, let $G$ be a graph, let $H$ be the underlying graph of the digraph $\delta G$, and let $n$ be an integer.
Then, $H$ admits a homomorphism to $\calO(\Fset,n)$ if and only if $G$ admits a homomorphism to $\calS(\Fset,n)$.
\end{lemma}

\begin{proof}
Put $G=(V_G,E_G)$ and $H=(V_H,E_H)$. Recall that the vertices of $H$, just like the vertices of $\delta G$, are the ordered pairs $(x,y)$ of adjacent vertices $x,y$ in $G$.

Suppose first that there exists a homomorphism $h$ from $H$ to $\calO(\Fset,n)$.
Consider the function $g$ that maps every vertex $y \in V_G$ to the subspace $g(y)$, spanned by the image of $h$ on the vertices of $H$ whose head is $y$, namely,
\[g(y) = \linspan ( \{ h(v) \mid v=(x,y) \in V_H~\mbox{for some}~x \in V_G\} ).\]
We claim that $g$ forms a homomorphism from $G$ to $\calS(\Fset,n)$.
Clearly, $g$ maps every vertex of $V_G$ to a subspace of $\Fset^n$, and thus to a vertex of $\calS(\Fset,n)$.
Further, let $x,y \in V_G$ be adjacent vertices in $G$, and consider the vector $w = h(x,y)$ assigned by $h$ to the vertex $(x,y)$ of $H$.
Since $w$ is a vertex of $\calO(\Fset,n)$, it holds that $\langle w,w \rangle \neq 0$. Since $(x,y)$ is a vertex of $H$ whose head is $y$, it follows that $w \in g(y)$. Further, every vertex of $H$ of the form $(x',x)$ for some $x' \in V_G$ is adjacent in $H$ to $(x,y)$, hence, since $h$ is a homomorphism to $\calO(\Fset,n)$, it holds that $\langle h(x',x),w \rangle = 0$. Since the subspace $g(x)$ is spanned by those vectors $h(x',x)$, we obtain that $w$ is orthogonal to the entire subspace $g(x)$. It thus follows that the vector $w$ satisfies $\langle w,w \rangle \neq 0$ and $w \in g(y) \cap g(x)^{\perp}$. By symmetry, there also exists a vector $w' \in \Fset^n$ satisfying $\langle w',w' \rangle \neq 0$ and $w' \in g(x) \cap g(y)^{\perp}$, hence the subspaces $g(x)$ and $g(y)$ are adjacent vertices in $\calS(\Fset,n)$, as required.

For the other direction, suppose that there exists a homomorphism $g$ from $G$ to $\calS(\Fset,n)$.
Consider the function $h$ that maps every vertex $(x,y) \in V_H$ of $H$ to some non-self-orthogonal vector $h(x,y)$ that lies in the intersection $g(x) \cap g(y)^{\perp}$.
Note that such a vector exists, because $x$ and $y$ are adjacent in $G$, hence $g(x)$ and $g(y)$ are adjacent in $\calS(\Fset,n)$.
We claim that $h$ forms a homomorphism from $H$ to $\calO(\Fset,n)$.
Clearly, $h$ maps every vertex of $V_H$ to a non-self-orthogonal vector of $\Fset^n$, and thus to a vertex of $\calO(\Fset,n)$.
Further, for every two adjacent vertices $(x,y),(y,z) \in V_H$ of $H$, it holds that $h(x,y) \in g(y)^{\perp}$ and $h(y,z) \in g(y)$, and thus $\langle h(x,y), h(y,z) \rangle = 0$. This implies that $h(x,y)$ and $h(y,z)$ are adjacent in $\calO(\Fset,n)$, so we are done.
\end{proof}

We will use Lemma~\ref{lemma:od_delta_new} to obtain lower bounds on the orthogonality dimension of the underlying graphs of line digraphs. To this end, we need upper bounds on the chromatic numbers of the graphs $\calS(\Fset,n)$.
Every vertex of $\calS(\Fset,n)$ is a subspace of $\Fset^n$ and thus can be represented by a basis that generates it.
For a finite field $\Fset$, the number of possible bases does not exceed $|\Fset|^{n^2}$, which obviously yields that $\chi(\calS(\Fset,n)) \leq |\Fset|^{n^2}$. While this simple bound suffices for proving our hardness results for the orthogonality dimension over finite fields, we note that the number of vertices in $\calS(\Fset,n)$ is in fact $|\Fset|^{(1+o(1)) \cdot n^2/4}$, where the $o(1)$ term tends to $0$ when $n$ tends to infinity.\footnote{To see this, put $q=|\Fset|$, and observe that the number of $k$-dimensional subspaces of $\Fset^n$ is precisely
$\prod_{i=0}^{k-1}{\tfrac{q^n-q^i}{q^k-q^i}}$ and that every term in this product lies in $[q^{n-k-1},q^{n-k+1}]$. Hence, the total number of subspaces of $\Fset^n$ is at least $\sum_{k=0}^{n}{q^{(n-k-1)k}}$ and at most $\sum_{k=0}^{n}{q^{(n-k+1)k}}$. It follows that the number of subspaces of $\Fset^n$ is $q^{(1+o(1)) \cdot n^2/4}$.}

We conclude this discussion with the following theorem.

\begin{theorem}\label{thm:od(H)}
Let $\Fset$ be a finite field, let $G$ be a graph, and let $H$ be the underlying graph of the digraph $\delta G$.
Then, it holds that
\[ \od_\Fset (H) \geq \sqrt{ \log_{|\Fset|} \chi(G)}.\]
\end{theorem}

\begin{proof}
Put $n = \od_\Fset(H)$.
By Proposition~\ref{prop:od_def}, $H$ admits a homomorphism to $\calO(\Fset,n)$, hence by Lemma~\ref{lemma:od_delta_new}, $G$ admits a homomorphism to $\calS(\Fset,n)$.
Since the chromatic number is monotone under homomorphisms, it follows that $\chi(G) \leq \chi(\calS(\Fset,n)) \leq |\Fset|^{n^2}$.
By rearranging, the proof is completed.
\end{proof}

\subsubsection{The Chromatic Number of \texorpdfstring{$\calS(\R,n)$}{S(R,n)}}\label{sec:S1_R}

For the real field $\R$ and for $n \geq 2$, the vertex set of the graph $\calS(\R,n)$ is infinite, and yet, its chromatic number is finite.
To see this, let us firstly observe a simple upper bound of $2^{3^n}$.
To each vertex of $\calS(\R,n)$, i.e., a subspace $U$ of $\R^n$, assign the subset of $\{0,\pm 1\}^n$ that consists of all the sign vectors of the vectors of $U$. This assignment forms a proper coloring of the graph, because for adjacent vertices $U$ and $V$ there exists a nonzero vector $w \in U$ that is orthogonal to $V$, hence the sign vector of $w$ belongs to the set of sign vectors of $U$ but does not belong to the one of $V$ (because the inner product of two vectors with the same nonzero sign vector is positive). Since the number of subsets of $\{0,\pm 1\}^n$ is $2^{3^n}$, it follows that $\chi(\calS(\R,n)) \leq 2^{3^n}$.

The above double-exponential bound is not sufficient for deriving $\NP$-hardness of approximation results for the orthogonality dimension over $\R$ from the currently known $\NP$-hardness results of the chromatic number. We therefore need the following lemma that provides an exponentially better bound which is suitable for our purposes.
For a vector $w \in \R^n$, we use here the notation $\|w\| = \sqrt{\langle w,w \rangle}$ for the Euclidean norm of $w$.

\begin{lemma}\label{lemma:chrom_S1}
For every integer $n$, it holds that $\chi(\calS(\R,n)) \leq (2n+1)^{n^2}$.
\end{lemma}

\begin{proof}
We define a coloring of the vertices of the graph $\calS(\R,n)$ as follows.
For every vertex of $\calS(\R,n)$, i.e., a subspace $U$ of $\R^n$, let $(u_1,\ldots,u_k)$ be an arbitrary orthonormal basis of $U$ where $k \leq n$, and assign $U$ to the color $c(U) = (u'_1, \ldots, u'_k)$ where $u'_i$ is a vector obtained from $u_i$ by rounding each of its values to a closest integer multiple of $\frac{1}{n}$. Note that for every $i \in [k]$, the vectors $u_i$ and $u'_i$ differ in every coordinate by no more than $\frac{1}{2n}$ in absolute value.

We claim that $c$ is a proper coloring of $\calS(\R,n)$.
To see this, let $U$ and $V$ be adjacent vertices in the graph.
If $\dim(U) \neq \dim(V)$ then it clearly holds that $c(U) \neq c(V)$. So suppose that the dimensions of $U$ and $V$ are equal, and put $k = \dim(U) = \dim(V)$.
Denote the orthonormal bases associated with $U$ and $V$ by $(u_1, \ldots,u_k)$ and $(v_1, \ldots,v_k)$ respectively, and let $c(U) = (u'_1, \ldots, u'_k)$ and $c(V) = (v'_1, \ldots, v'_k)$ be their colors. Our goal is to show that $c(U) \neq c(V)$.

Assume for the sake of contradiction that $c(U) = c(V)$, that is, $u'_i = v'_i$ for every $i \in [k]$.
This implies that for every $i \in [k]$, the vectors $u_i$ and $v_i$ differ in each coordinate by no more than $\frac{1}{n}$ in absolute value, hence
\begin{eqnarray}\label{eq:ui-vi}
\|u_i - v_i\| \leq \sqrt{n \cdot \frac{1}{n^2}} = \frac{1}{\sqrt{n}}.
\end{eqnarray}
Since $U$ and $V$ are adjacent in the graph $\calS(\R,n)$, by scaling, there exists a unit vector $u \in U \cap V^{\perp}$.
Write $u = \sum_{i \in [k]}{\alpha_i \cdot u_i}$ for coefficients $\alpha_1,\ldots,\alpha_k \in \R$.
Since the given basis of $U$ is orthonormal, it follows that $\sum_{i \in [k]}{\alpha_i^2} = \|u\|^2 = 1$.
Now, consider the vector $v = \sum_{i \in [k]}{\alpha_i \cdot v_i}$, and observe that $v$ is a unit vector that belongs to the subspace $V$. Observe further that
\begin{eqnarray}\label{eq:u-v}
\|u-v\| = \Big \| \sum_{i \in [k]}{\alpha_i \cdot (u_i - v_i)} \Big \| \leq \sum_{i \in [k]}{|\alpha_i| \cdot \|u_i - v_i\|} \leq \Big (\sum_{i \in [k]}{\alpha_i^2} \Big )^{1/2} \cdot \Big ( \sum_{i \in [k]}{\|u_i-v_i\|^2} \Big )^{1/2} \leq 1,
\end{eqnarray}
where the first inequality follows from the triangle inequality, the second from the Cauchy-Schwarz inequality, and the third from~\eqref{eq:ui-vi} using $k \leq n$.
However, $u$ and $v$ are orthogonal unit vectors, and as such, the distance between them satisfies $\|u - v \| = \sqrt{2}$. This yields a contradiction to~\eqref{eq:u-v}, hence $c(U) \neq c(V)$.

To complete the proof, we observe that the number of colors used by the proper coloring $c$ does not exceed $(2n+1)^{n^2}$. Indeed, every color can be represented by an $n \times n$ matrix whose values are of the form $\frac{a}{n}$ for integers $-n \leq a \leq n$ (where the matrix associated with a subspace of dimension $k$ consists of the rounded $k$ column vectors concatenated with $n-k$ columns of zeros). Since the number of those matrices is bounded by $(2n+1)^{n^2}$, we are done.
\end{proof}

We derive the following theorem.

\begin{theorem}\label{thm:od_R(H)}
There exists a constant $c>0$, such that for every graph $G$ with $\chi(G) \geq 3$, the underlying graph $H$ of the digraph $\delta G$ satisfies
\[ \od_\R (H) \geq c \cdot \sqrt{\tfrac{\log \chi(G)}{\log \log \chi(G)}}.\]
\end{theorem}

\begin{proof}
Put $n = \od_\R(H)$.
By Proposition~\ref{prop:od_def}, $H$ admits a homomorphism to $\calO(\R,n)$, hence by Lemma~\ref{lemma:od_delta_new}, $G$ admits a homomorphism to $\calS(\R,n)$.
Using Lemma~\ref{lemma:chrom_S1}, we obtain that
\[\chi(G) \leq \chi(\calS(\R,n)) \leq (2n+1)^{n^2},\]
which yields the desired bound.
\end{proof}

\subsubsection{The Clique Number of \texorpdfstring{$\calS(\Fset,n)$}{S(F,n)}}\label{sec:clique_S1}

We next consider the clique numbers of the graphs $\calS(\Fset,n)$, whose estimation is motivated by the following lemma.
Here, the clique number of a graph $G$ is denoted by $\omega(G)$.
\begin{lemma}\label{lemma:omega_motiv}
Let $\Fset$ be a field, let $G$ be a graph, and let $H$ be the underlying graph of the digraph $\delta G$.
If $\chi(G) \leq \omega(\calS(\Fset,n))$, then $\od_\Fset(H) \leq n$.
\end{lemma}
\begin{proof}
Put $m = \omega(\calS(\Fset,n))$, and let $U_1, \ldots, U_m$ be $m$ subspaces of $\Fset^n$ that form a clique in $\calS(\Fset,n)$.
Put $G=(V,E)$, suppose that $\chi(G) \leq m$, and let $c: V \rightarrow [m]$ be a proper coloring  of $G$.
Notice that for every two adjacent vertices $x,y$ in $G$, the subspaces $U_{c(x)}$ and $U_{c(y)}$ are adjacent vertices in $\calS(\Fset,n)$.

We define an $n$-dimensional orthogonal representation of $H$ over $\Fset$ as follows.
Recall that every vertex of $H$ is a pair $(x,y)$ of adjacent vertices $x,y$ in $G$. Assign every such vertex $(x,y)$ to some non-self-orthogonal vector $u_{(x,y)}$ that lies in $U_{c(y)} \cap U_{c(x)}^{\perp}$. The existence of such a vector follows from the adjacency of the vertices $U_{c(x)}$ and $U_{c(y)}$ in $\calS(\Fset,n)$.
We claim that this assignment is an orthogonal representation of $H$. Indeed, for adjacent vertices $(x,y)$ and $(y,z)$ in $H$, the vector $u_{(x,y)}$ belongs to $U_{c(y)}$ whereas the vector $u_{(y,z)}$ is orthogonal to $U_{c(y)}$, hence they satisfy $\langle u_{(x,y)}, u_{(y,z)} \rangle = 0$. Since this orthogonal representation lies in $\Fset^n$, we establish that $\od_\Fset(H) \leq n$.
\end{proof}

For a graph $G$ and for the underlying graph $H$ of its line digraph $\delta G$, Theorem~\ref{thm:chi_delta} implies that if $\chi(G) \leq \binom{n}{\lfloor n/2 \rfloor}$ then $\chi(H) \leq n$, and thus, by Claim~\ref{claim:mr_od_chi}, $\od_\Fset(H) \leq n$ for every field $\Fset$. This raises the question of whether Lemma~\ref{lemma:omega_motiv} can be used to obtain a better upper bound on $\od_\Fset(H)$ as a function of $\chi(G)$. For certain cases, the following result answers this question negatively. Namely, it shows that the clique number of the graphs $\calS(\Fset,n)$ is precisely $\binom{n}{\lfloor n/2 \rfloor}$, whenever the vector space $\Fset^n$ has no nonzero self-orthogonal vectors (as in the case of $\Fset = \R$). It thus follows that Lemma~\ref{lemma:omega_motiv} cannot yield a better relation between the quantities $\od_\R(H)$ and $\chi(G)$ than the one stemming from Theorem~\ref{thm:chi_delta}.

\begin{proposition}\label{prop:clique_S1_R}
For a field $\Fset$ and an integer $n$ such that $\Fset^n$ has no nonzero self-orthogonal vectors, it holds that \[\omega (\calS(\Fset,n)) = \binom{n}{\lfloor n/2 \rfloor }.\]
\end{proposition}

The proof of Proposition~\ref{prop:clique_S1_R} relies on the following result of Kalai~\cite{Kalai80} (see also~\cite{LovaszFlats77}).

\begin{theorem}[\cite{Kalai80}]\label{thm:Kalai}
For a field $\Fset$ and an integer $n$, let $(U_1,W_1), \ldots, (U_m,W_m)$ be $m$ pairs of subspaces of $\Fset^n$ such that
\begin{enumerate}
  \item $U_i \cap W_i = \{0\}$ for every $i \in [m]$, and
  \item $U_i \cap W_j \neq \{0\}$ for every $i \neq j \in [m]$.
\end{enumerate}
Then, $m \leq \binom{n}{\lfloor n/2 \rfloor }$.
\end{theorem}

\begin{proof}[ of Proposition~\ref{prop:clique_S1_R}]
We first show that there exists a clique in $\calS(\Fset,n)$ of size $\binom{n}{\lfloor n/2 \rfloor }$.
For every set $A \subseteq [n]$ of size $|A| = \lfloor n/2 \rfloor$, let $U_A$ denote the subspace of $\Fset^n$ spanned by the vectors $e_i$ with $i \in A$, where $e_i$ stands for the vector of $\Fset^n$ with $1$ on the $i$th entry and $0$ everywhere else. It clearly holds that for every distinct such sets $A_1,A_2$, there exists some $i \in A_1 \setminus A_2$, and that the vector $e_i$ satisfies $\langle e_i,e_i \rangle =1$ and $e_i \in U_{A_1} \cap U_{A_2}^{\perp}$. It thus follows that the $\binom{n}{\lfloor n/2 \rfloor }$ subspaces $U_A$ with $|A| = \lfloor n/2 \rfloor$ form a clique in the graph $\calS(\Fset,n)$, as required.

We next show that the size of every clique in $\calS(\Fset,n)$ does not exceed $\binom{n}{\lfloor n/2 \rfloor }$.
To see this, let $U_1, \ldots, U_m$ be subspaces of $\Fset^n$ that form a clique in $\calS(\Fset,n)$.
Consider the pairs $(U_i, U_i^{\perp})$ for $i \in [m]$, and observe that they satisfy the conditions of Theorem~\ref{thm:Kalai}. Indeed, for every $i \in [m]$ it holds that $U_i \cap U_i^{\perp} = \{0\}$, because $\Fset^n$ has no nonzero self-orthogonal vectors. Further, since the given collection of subspaces is a clique in $\calS(\Fset,n)$, for every $i \neq j \in [m]$, there exists a vector $w \in \Fset^n$ with $\langle w,w \rangle \neq 0$ such that $w \in U_i \cap U_j^{\perp}$, hence, $U_i \cap U_j^{\perp} \neq \{0\}$. It thus follows from Theorem~\ref{thm:Kalai} that $m \leq \binom{n}{\lfloor n/2 \rfloor }$, as required.
\end{proof}

\subsection{Minrank}\label{sec:minrk_line}

As in the previous section, we start with a definition of a family of graphs.
\begin{definition}\label{def:S2}
For a field $\Fset$ and an integer $n$, let $\calS'(\Fset,n)$ denote the graph whose vertices are all the pairs of subspaces of $\Fset^n$, where two distinct pairs $(U_1,W_1)$ and $(U_2,W_2)$ are adjacent if there exist two vectors $u,w \in \Fset^n$ with $\langle u,w \rangle \neq 0$ such that $u \in U_1 \cap W_2^{\perp}$ and $w \in W_1 \cap U_2^{\perp}$ and, in addition, there exist two vectors $u',w' \in \Fset^n$ with $\langle u',w' \rangle \neq 0$ such that $u' \in U_2 \cap W_1^{\perp}$ and $w' \in W_2 \cap U_1^{\perp}$.
\end{definition}

We next argue that the chromatic number of a graph $G$ can be used to estimate the minrank of the complement of the underlying graph of its line digraph $\delta G$.
This is established using the following lemma that involves the graphs $\calS'(\Fset,n)$.
Its proof resembles that of Lemma~\ref{lemma:od_delta_new}.

\begin{lemma}\label{lemma:mr_delta_new}
Let $\Fset$ be a field, let $G$ be a graph, let $H$ be the underlying graph of the digraph $\delta G$, and let $n$ be an integer.
Then, $H$ admits a homomorphism to $\calO'(\Fset,n)$ if and only if $G$ admits a homomorphism to $\calS'(\Fset,n)$.
\end{lemma}

\begin{proof}
Put $G=(V_G,E_G)$ and $H=(V_H,E_H)$.
Suppose first that there exists a homomorphism $h$ from $H$ to $\calO'(\Fset,n)$.
For every vertex $y \in V_G$ of $G$, let $U_y$ denote the subspace of $\Fset^n$ spanned by the vectors $u$ of the pairs $(u,w)$ that lie in the image of $h$ on the vertices of $H$ whose head is $y$, namely,
\[U_y = \linspan ( \{ u \mid (u,w)= h(x,y) ~\mbox{for some}~(x,y) \in V_H~\mbox{and}~w \in \Fset^n\} ).\]
Similarly, let $W_y$ denote the subspace of $\Fset^n$ spanned by the vectors $w$ of the pairs $(u,w)$ that lie in the image of $h$ on the vertices of $H$ whose head is $y$, namely,
\[W_y = \linspan ( \{ w \mid (u,w)= h(x,y) ~\mbox{for some}~(x,y) \in V_H~\mbox{and}~u \in \Fset^n\} ).\]
Consider the function $g$ that maps every vertex $y \in V_G$ of $G$ to the pair $g(y) = (U_y,W_y)$.
We claim that $g$ forms a homomorphism from $G$ to $\calS'(\Fset,n)$.
Clearly, $g$ maps every vertex of $V_G$ to a pair of subspaces of $\Fset^n$, and thus to a vertex of $\calS'(\Fset,n)$.
Further, let $x,y \in V_G$ be adjacent vertices in $G$, and consider the pair $(u,w) = h(x,y)$ assigned by $h$ to the vertex $(x,y)$ of $H$.
Since $(u,w)$ is a vertex of $\calO'(\Fset,n)$, it holds that $\langle u,w \rangle \neq 0$. Since $(x,y)$ is a vertex of $H$ whose head is $y$, it follows that $u \in U_y$ and $w \in W_y$. Further, every vertex of $H$ of the form $(x',x)$ for some $x' \in V_G$ is adjacent in $H$ to $(x,y)$. Put $(\tilde{u},\tilde{w}) = h(x',x)$, and observe that the fact that $h$ is a homomorphism to the graph $\calO'(\Fset,n)$ implies that $\langle \tilde{u},w \rangle = \langle u,\tilde{w} \rangle = 0$. Since the subspaces $U_x$ and $W_x$ are spanned, respectively, by those vectors $\tilde{u}$ and $\tilde{w}$, we obtain that $u$ is orthogonal to the subspace $W_x$ and $w$ is orthogonal to the subspace $U_x$. It thus follows that the vectors $u$ and $w$ satisfy $\langle u,w \rangle \neq 0$, $u \in U_y \cap W_x^{\perp}$, and $w \in W_y \cap U_x^{\perp}$. By symmetry, there also exist vectors $u',w' \in \Fset^n$ satisfying $\langle u',w' \rangle \neq 0$, $u' \in U_x \cap W_y^{\perp}$, and $w' \in W_x \cap U_y^{\perp}$, hence the pairs $g(x)=(U_x,W_x)$ and $g(y)=(U_y,W_y)$ are adjacent vertices in $\calS'(\Fset,n)$, as required.

For the other direction, suppose that there exists a homomorphism $g$ from $G$ to $\calS'(\Fset,n)$.
Consider the function $h$ defined as follows.
For every vertex $(x,y) \in V_H$ of $H$, consider the pairs $g(x) = (U_x,W_x)$ and $g(y) = (U_y,W_y)$, and let $h(x,y)$ be a pair $(u,w)$, where $u$ and $w$ are vectors satisfying $\langle u,w \rangle \neq 0$, $u \in U_x \cap W_y^{\perp}$, and $w \in W_x \cap U_y^{\perp}$.
Note that such a pair exists, because $x$ and $y$ are adjacent in $G$, hence $g(x)$ and $g(y)$ are adjacent in $\calS'(\Fset,n)$.
We claim that $h$ forms a homomorphism from $H$ to $\calO'(\Fset,n)$.
Clearly, $h$ maps every vertex of $V_H$ to a pair of non-orthogonal vectors of $\Fset^n$, and thus to a vertex of $\calO'(\Fset,n)$.
Further, let $(x,y),(y,z) \in V_H$ be two adjacent vertices of $H$, and put $g(y) = (U_y,W_y)$. Notice that the first vector of $h(x,y)$ lies in $W_y^{\perp}$ and that the second vector of $h(y,z)$ lies in $W_y$, hence they are orthogonal. Similarly, the second vector of $h(x,y)$ lies in $U_y^{\perp}$ and the first vector of $h(y,z)$ lies in $U_y$, hence they are orthogonal too. This implies that $h(x,y)$ and $h(y,z)$ are adjacent in $\calO'(\Fset,n)$, so we are done.
\end{proof}

We derive the following theorem.

\begin{theorem}\label{thm:mr(H)}
Let $\Fset$ be a finite field, let $G$ be a graph, and let $H$ be the underlying graph of the digraph $\delta G$.
Then, it holds that
\[ {\minrank}_\Fset(\overline{H}) \geq \sqrt{ \tfrac{1}{2} \cdot \log_{|\Fset|} \chi(G)}.\]
\end{theorem}

\begin{proof}
Put $n = {\minrank}_\Fset(\overline{H})$.
By Proposition~\ref{prop:minrk_def}, $H$ admits a homomorphism to $\calO'(\Fset,n)$, hence by Lemma~\ref{lemma:mr_delta_new}, $G$ admits a homomorphism to $\calS'(\Fset,n)$.
Since the chromatic number is monotone under homomorphisms, it follows that
\[\chi(G) \leq \chi(\calS'(\Fset,n)) \leq |\Fset|^{2n^2},\]
where the second inequality holds because the number of vertices in $\calS'(\Fset,n)$ does not exceed $|\Fset|^{2n^2}$.
By rearranging, the proof is completed.
\end{proof}

\subsubsection{The Chromatic Number of \texorpdfstring{$\calS'(\R,n)$}{S'(R,n)}}\label{sec:S2_R}

We next consider the problem of determining the chromatic numbers of the graphs $\calS'(\R,n)$.
The following theorem shows that these graphs cannot be properly colored using a finite number of colors, in contrast to the graphs $\calS(\R,n)$ addressed in Lemma~\ref{lemma:chrom_S1}.

\begin{theorem}\label{thm:S2_infty}
For every integer $n \geq 3$, it holds that $\chi(\calS'(\R,n)) = \infty$.
\end{theorem}

Before proving Theorem~\ref{thm:S2_infty}, let us describe a significant difference between the behavior of $\od_\R(G)$ and of ${\minrank}_\R(\overline{G})$ with respect to the chromatic number $\chi(G)$. It is not difficult to see that the chromatic number of a graph $G$ is bounded from above by some function of $\od_\R(G)$. Indeed, the graph $\calO(\R,n)$ is $3^n$-colorable, as follows from the coloring that assigns to every nonzero vector of $\R^n$ its sign vector from $\{0,\pm1\}^n$. By Proposition~\ref{prop:od_def}, this implies that every graph $G$ satisfies $\chi(G) \leq 3^{\od_\R(G)}$ (see also~\cite[Chapter~11]{LovaszBook}).
On the other hand, the chromatic number of a graph $G$ cannot be bounded from above by any function of ${\minrank}_\R(\overline{G})$, as proved below.

\begin{theorem}\label{thm:minrk3}
For every integer $m$, there exists a graph $G$ such that ${\minrank}_\R(\overline{G}) \leq 3$ and yet $\chi(G) \geq m$.
\end{theorem}

\begin{proof}
For an integer $n>6$, consider the `double shift graph' $G_n$ defined as follows. Its vertices are all the $3$-subsets of $[n]$, where two sets $\{x_1,x_2,x_3\}$ and $\{y_1,y_2,y_3\}$ with $x_1<x_2<x_3$ and $y_1<y_2<y_3$ are adjacent in $G_n$ if either $(x_2,x_3) = (y_1,y_2)$ or $(x_1,x_2) = (y_2,y_3)$. It was shown in~\cite{ErdosH66} that the graph $G_n$ satisfies $\chi(G_n) = (1+o(1)) \cdot \log \log n$ (see also~\cite{FurediHRT91}), whereas its local chromatic number, a concept introduced by Erd{\"{o}}s et al.~\cite{ErdosLocal}, is known to be $3$. By an argument of Shanmugam, Dimakis, and Langberg~\cite[Theorem~1]{SDLlocal13}, this implies that ${\minrank}_\R(\overline{G_n}) \leq 3$ (see also~\cite[Proposition~6.5]{AttiasH21}).
This completes the proof.
\end{proof}

We are ready to derive Theorem~\ref{thm:S2_infty}.

\begin{proof}[ of Theorem~\ref{thm:S2_infty}]
It clearly suffices to prove the assertion of the theorem for $n=3$.
Consider the subgraph of $\calS'(\R,3)$ induced by the pairs $(U,W)$ of subspaces of $\R^3$ such that $U$ and $W$ are non-orthogonal subspaces of dimension $1$, and observe that $\calO'(\R,3)$ admits a homomorphism to this subgraph.
By Proposition~\ref{prop:minrk_def}, for every graph $G$ with ${\minrank}_\R(\overline{G}) \leq 3$, there exists a homomorphism from $G$ to $\calO'(\R,3)$ and thus $\chi(G) \leq \chi(\calO'(\R,3))$.
By Theorem~\ref{thm:minrk3}, the chromatic number of a graph $G$ with ${\minrank}_\R(\overline{G}) \leq 3$ can be arbitrarily large, hence $\chi(\calO'(\R,3)) = \infty$. Since $\calO'(\R,3)$ admits a homomorphism to a subgraph of $\calS'(\R,3)$, this yields that $\chi(\calS'(\R,3)) = \infty$, as required.
\end{proof}

\subsection{Index Coding}

In this section, we study the optimal length of (not necessarily linear) index codes for the complement of underlying graphs of line digraphs. Recall Definition~\ref{def:index}.

We start by presenting an argument of Langberg and Sprintson~\cite[Theorem~4(a)]{LangbergS08} that relates the chromatic number of a graph to the length of an index code for its complement.
In fact, we slightly modify their argument to obtain the improved bound stated below (with $2^{|\Sigma|^k}$ rather than $|\Sigma|^{|\Sigma|^k}$ in the statement of the result).

\begin{proposition}\label{prop:index}
Let $\Sigma$ be an alphabet of size at least $2$, and let $G$ be a graph.
If there exists an index code for $\overline{G}$ over $\Sigma$ of length $k$, then $\chi(G) \leq 2^{|\Sigma|^k}$.
\end{proposition}

\begin{proof}
Assume without loss of generality that $\{0,1\} \subseteq \Sigma$.
Put $G=(V,E)$ and $n = |V|$. Suppose that there exists an index code for $\overline{G}$ over $\Sigma$ of length $k$, and let $E:\Sigma^n \rightarrow \Sigma^k$ and $g_i :\Sigma^{k+|N_{\overline{G}}(i)|} \rightarrow  \Sigma$ for $i \in V$ denote the corresponding encoding and decoding functions.

For every vertex $i \in V$, we define a function $h_i: \Sigma^k \rightarrow \{0,1\}$ that determines for a given encoded message $y \in \Sigma^k$ whether $g_i$ returns $0$ on $y$ when all the symbols of the side information of the $i$th receiver are zeros. Formally speaking, for every $y \in \Sigma^k$, we define $h_i(y)=0$ if $g_i(y,0,\ldots,0)=0$, and $h_i(y)=1$ otherwise.

We claim that the assignment of the function $h_i$ to each vertex $i \in V$ forms a proper coloring of $G$.
To see this, let $i$ and $j$ be adjacent vertices in $G$.
Let $x \in \Sigma^n$ denote the vector with $1$ in the $i$th entry and $0$ everywhere else, and put $y = E(x)$.
By the correctness of the decoding functions, it follows that $g_i(y,x|_{N_{\overline{G}}(i)})= x_i = 1$ whereas $g_j(y,x|_{N_{\overline{G}}(j)})= x_j = 0$.
Since $i$ and $j$ are adjacent in $G$, they are not adjacent in $\overline{G}$, hence all the symbols in the side information $x|_{N_{\overline{G}}(i)}$ of $i$ and in the side information $x|_{N_{\overline{G}}(j)}$ of $j$ are zeros. This implies that $g_i(y,0,\ldots,0)=1$ and $g_j(y,0,\ldots,0)=0$, and therefore $h_i(y) = 1$ and $h_j(y) = 0$, which yields that $h_i \neq h_j$, as required.
Finally, observe that the number of distinct functions $h_i: \Sigma^k \rightarrow \{0,1\}$ for $i \in V$ does not exceed $2^{|\Sigma|^k}$, implying that $\chi(G) \leq 2^{|\Sigma|^k}$.
\end{proof}

We proceed by proving an analogue of Proposition~\ref{prop:index} for line digraphs.

\begin{theorem}\label{thm:index_b}
Let $\Sigma$ be an alphabet of size at least $2$, let $G$ be a graph, and let $H$ be the underlying graph of the digraph $\delta G$.
If there exists an index code for $\overline{H}$ over $\Sigma$ of length $k$, then $\chi(G) \leq 2^{|\Sigma|^k}$.
\end{theorem}

\begin{proof}
Assume without loss of generality that $\{0,1\} \subseteq \Sigma$.
Put $G=(V_G,E_G)$, $H = (V_H,E_H)$, and $n = |V_H|$.
Recall that the vertices of $H$ are the ordered pairs of adjacent vertices in $G$, hence $n = 2 \cdot |E_G|$.
Suppose that there exists an index code for $\overline{H}$ over $\Sigma$ of length $k$, and let $E:\Sigma^n \rightarrow \Sigma^k$ and $g_{(u,v)} :\Sigma^{k+|N_{\overline{H}}(u,v)|} \rightarrow  \Sigma$ for $(u,v) \in V_H$ denote the corresponding encoding and decoding functions.

For every vertex $v \in V_G$, we define a function $h_v: \Sigma^k \rightarrow \{ 0,1\}$ that determines for a given encoded message $y \in \Sigma^k$ whether every function $g_{(u,v)}$ associated with a vertex $(u,v) \in V_H$ returns $0$ on $y$ when all the symbols in the side information of the receiver of the vertex $(u,v)$ are zeros. Formally speaking, for every $y \in \Sigma^k$, we define $h_v(y)=0$ if for every $u \in V_G$ with $(u,v) \in V_H$, it holds that $g_{(u,v)}(y,0,\ldots,0)=0$, and $h_v(y)=1$ otherwise.

We claim that the assignment of the function $h_v$ to each vertex $v \in V_G$ forms a proper coloring of $G$.
To see this, let $v_1$ and $v_2$ be adjacent vertices in $G$, and notice that $(v_1,v_2)$ is a vertex of $H$.
Let $x \in \Sigma^n$ denote the vector with $1$ in the entry of $(v_1,v_2)$ and $0$ everywhere else, and put $y = E(x)$.

We first claim that $h_{v_1}(y) = 0$.
To see this, consider any vertex $(u,v_1) \in V_H$, and notice that $(u,v_1)$ and $(v_1,v_2)$ are adjacent in $H$ and are thus not adjacent in $\overline{H}$.
By the correctness of the decoding function $g_{(u,v_1)}$, it follows that $g_{(u,v_1)}(y,x|_{N_{\overline{H}}(u,v_1)})= x_{(u,v_1)} = 0$.
Since $(u,v_1)$ and $(v_1,v_2)$ are not adjacent in $\overline{H}$, all the symbols in the side information $x|_{N_{\overline{H}}(u,v_1)}$ of the vertex $(u,v_1)$ are zeros.
We thus obtain that for every vertex $u \in V_G$ with $(u,v_1) \in V_H$, it holds that $g_{(u,v_1)}(y,0,\ldots,0) = 0$.
By the definition of $h_{v_1}$, it follows that $h_{v_1}(y) = 0$, as required.

We next claim that $h_{v_2}(y) = 1$.
To see this, observe that by the correctness of the decoding function $g_{(v_1,v_2)}$, it follows that $g_{(v_1,v_2)}(y,x|_{N_{\overline{H}}(v_1,v_2)})= x_{(v_1,v_2)} = 1$.
It further holds that all the symbols in the side information $x|_{N_{\overline{H}}(v_1,v_2)}$ of the vertex $(v_1,v_2)$ are zeros. By the definition of $h_{v_2}$, it follows that $h_{v_2}(y) = 1$, as required.

We obtain that every two adjacent vertices $v_1$ and $v_2$ in $G$ satisfy $h_{v_1} \neq h_{v_2}$.
Since the number of functions $h_v: \Sigma^k \rightarrow \{ 0,1\}$ for $v \in V_G$ does not exceed $2^{|\Sigma|^k}$, it follows that $\chi(G) \leq 2^{|\Sigma|^k}$, and we are done.
\end{proof}

\section{Hardness Results}\label{sec:hard}

In this section, we prove our hardness results for the promise problems associated with orthogonality dimension and minrank, which imply hardness of approximation for these quantities.
We also suggest a potential avenue for proving hardness results for the general index coding problem over a constant-size alphabet.

The starting point of our hardness proofs is the following theorem of Wrochna and {\v Z}ivn{\' y}~\cite{WZ20}.
Recall that the function $b: \N \rightarrow \N$ is defined by $b(n) = \binom{n}{\lfloor n/2 \rfloor}$.

\begin{theorem}[\cite{WZ20}]\label{thm:ColHard}
For every integer $k \geq 4$, it is $\NP$-hard to decide whether a given graph $G$ satisfies $\chi(G) \leq k$ or $\chi(G) \geq b(k)$.
\end{theorem}

Our hardness results for the orthogonality dimension and the minrank parameter over finite fields are given by the following theorem, which confirms Theorem~\ref{thm:IntroF}.

\begin{theorem}\label{thm:hardnessF}
There exists a function $f:\N \rightarrow \N$ satisfying $f(k) = (1-o(1)) \cdot \sqrt{b(k)}$ such that for every finite field $\Fset$ and for every sufficiently large integer $k$, the following holds.
\begin{enumerate}
  \item It is $\NP$-hard to decide whether a given graph $G$ satisfies
  \[\od_\Fset(G) \leq k~~~\mbox{or}~~~\od_\Fset(G) \geq \tfrac{1}{\sqrt{\log |\Fset|}} \cdot f(k).\]
  \item It is $\NP$-hard to decide whether a given graph $G$ satisfies
  \[{\minrank}_\Fset(G) \leq k~~~\mbox{or}~~~{\minrank}_\Fset(G) \geq \tfrac{1}{\sqrt{2 \cdot \log |\Fset|}} \cdot f(k).\]
\end{enumerate}
\end{theorem}

\begin{proof}
Fix a finite field $\Fset$.
We start by proving the first item of the theorem.
For an integer $k \geq 4$, consider the problem of deciding whether a given graph $G$ satisfies
\[\chi(G) \leq b(k)~~~\mbox{or}~~~\chi(G) \geq b(b(k)),\]
whose $\NP$-hardness follows from Theorem~\ref{thm:ColHard}.
To obtain our hardness result on the orthogonality dimension over $\Fset$, we reduce from this problem.
Consider the reduction that given an input graph $G$ produces and outputs the underlying graph $H$ of the digraph $\delta G$.
This reduction can clearly be implemented in polynomial time (in fact, in logarithmic space).

To prove the correctness of the reduction, we analyze the orthogonality dimension of $H$ over $\Fset$.
If $G$ is a $\YES$ instance, that is, $\chi(G) \leq b(k)$, then by combining Claim~\ref{claim:mr_od_chi} with Theorem~\ref{thm:chi_delta}, it follows that
\[\od_\Fset(H) \leq \chi(H) \leq k.\]
If $G$ is a $\NO$ instance, that is, $\chi(G) \geq b(b(k))$, then by Theorem~\ref{thm:od(H)}, it follows that
\[ \od_\Fset(H) \geq \sqrt{\log_{|\Fset|} \chi(G)} \geq \sqrt{\log_{|\Fset|} b(b(k)) } = \tfrac{1-o(1)}{\sqrt{\log {|\Fset|}}} \cdot \sqrt{b(k)},\]
where the $o(1)$ term tends to $0$ when $k$ tends to infinity. Note that we have used here the fact that $b(n) = \Theta(2^n/\sqrt{n})$.
By letting $k$ be any sufficiently large integer, the proof of the first item of the theorem is completed.

The proof of the second item of the theorem is similar.
To avoid repetitions, we briefly mention the needed changes in the proof.
First, to obtain a hardness result for the minrank parameter, the reduction has to output the complement $\overline{H}$ of the graph $H$ rather than $H$ itself.
Second, in the analysis of the $\NO$ instances, one has to apply Theorem~\ref{thm:mr(H)} instead of Theorem~\ref{thm:od(H)} to obtain that
\[ {\minrank}_\Fset(\overline{H}) \geq \sqrt{\tfrac{1}{2} \cdot \log_{|\Fset|} \chi(G) } \geq \sqrt{\tfrac{1}{2} \cdot \log_{|\Fset|} b(b(k)) } = \tfrac{1-o(1)}{\sqrt{2 \cdot \log |\Fset|}} \cdot \sqrt{b(k)}.\]
This completes the proof of the theorem.
\end{proof}

As an immediate corollary of Theorem~\ref{thm:hardnessF}, we obtain the following.

\begin{corollary}
For every finite field $\Fset$, the following holds.
\begin{enumerate}
  \item It is $\NP$-hard to approximate $\od_\Fset(G)$ for a given graph $G$ to within any constant factor.
  \item It is $\NP$-hard to approximate ${\minrank}_\Fset(G)$ for a given graph $G$ to within any constant factor.
\end{enumerate}
\end{corollary}

We next prove a hardness result for the orthogonality dimension over the reals, confirming Theorem~\ref{thm:IntroR}.

\begin{theorem}\label{thm:hardnessR}
There exists a function $f:\N \rightarrow \N$ satisfying $f(k) = \Theta(\sqrt{b(k)/k})$ such that for every sufficiently large integer $k$, it is $\NP$-hard to decide whether a given graph $G$ satisfies
\[\od_\R(G) \leq k~~~\mbox{or}~~~\od_\R(G) \geq f(k).\]
\end{theorem}

\begin{proof}
As in the proof of Theorem~\ref{thm:hardnessF}, for an integer $k \geq 4$, we reduce from the problem of deciding whether a given graph $G$ satisfies
\[\chi(G) \leq b(k)~~~\mbox{or}~~~\chi(G) \geq b(b(k)),\]
whose $\NP$-hardness follows from Theorem~\ref{thm:ColHard}.
Consider the polynomial-time reduction that given an input graph $G$ produces and outputs the underlying graph $H$ of the digraph $\delta G$.

To prove the correctness of the reduction, we analyze the orthogonality dimension of $H$ over $\R$.
If $G$ is a $\YES$ instance, that is, $\chi(G) \leq b(k)$, then by combining Claim~\ref{claim:mr_od_chi} with Theorem~\ref{thm:chi_delta}, it follows that
\[\od_\R(H) \leq \chi(H) \leq k.\]
If $G$ is a $\NO$ instance, that is, $\chi(G) \geq b(b(k))$, then by Theorem~\ref{thm:od_R(H)} combined with the fact that $b(n) = \Theta(2^n/\sqrt{n})$, it follows that
\[ \od_\R(H) \geq c \cdot \sqrt{\tfrac{\log b(b(k))}{\log \log b(b(k))}} = \Theta \Big ( \sqrt{\tfrac{b(k)}{k}} \Big ),\]
where $c$ is an absolute positive constant.
This completes the proof of the theorem.
\end{proof}

As an immediate corollary of Theorem~\ref{thm:hardnessR}, we obtain the following.

\begin{corollary}
It is $\NP$-hard to approximate $\od_\R(G)$ for a given graph $G$ to within any constant factor.
\end{corollary}

We end this section with a statement that might be useful for proving $\NP$-hardness results for the general index coding problem.
Consider the following definition.
\begin{definition}
For an alphabet $\Sigma$ and for two integers $k_1 < k_2$, let $\IC_\Sigma (k_1,k_2)$ denote the problem of deciding whether the minimal length of an index code for a given graph $G$ over $\Sigma$ is at most $k_1$ or at least $k_2$.
\end{definition}

We prove the following result.

\begin{theorem}\label{thm:index}
Let $\Sigma$ be an alphabet of size at least $2$, and let $k_1,k_2$ be two integers.
Then, there exists a polynomial-time reduction from the problem of deciding whether a given graph $G$ satisfies $\chi(G) \leq b(k_1)$ or $\chi(G) \geq k_2$ to $\IC_\Sigma (k_1,\log_{|\Sigma|}\log k_2)$.
\end{theorem}

\begin{proof}
Consider the polynomial-time reduction that given an input graph $G$ produces the underlying graph $H$ of the digraph $\delta G$ and outputs its complement $\overline{H}$.
For correctness, suppose first that $G$ is a $\YES$ instance, that is, $\chi(G) \leq b(k_1)$. Then, by combining Claim~\ref{claim:mr_od_chi} with Theorem~\ref{thm:chi_delta}, it follows that ${\minrank}_{\Fset_2}(\overline{H}) \leq \chi(H) \leq k_1$.
By Proposition~\ref{prop:ic_minrk}, it further follows that there exists a linear index code for $\overline{H}$ over $\Fset_2$ of length $k_1$.
In particular, using $|\Sigma| \geq 2$, there exists an index code for $\overline{H}$ over the alphabet $\Sigma$ of length $k_1$.
Suppose next that $G$ is a $\NO$ instance, that is, $\chi(G) \geq k_2$. By Theorem~\ref{thm:index_b}, it follows that the length of any index code for $\overline{H}$ over $\Sigma$ is at least $\log_{|\Sigma|}\log k_2$, so we are done.
\end{proof}

Theorem~\ref{thm:index} implies that in order to prove the $\NP$-hardness of the general index coding problem over some finite alphabet $\Sigma$ of size at least $2$, it suffices to prove for some integer $k$ that it is $\NP$-hard to decide whether a given graph $G$ satisfies $\chi(G) \leq b(k)$ or $\chi(G) > 2^{|\Sigma|^k}$.

\section*{Acknowledgements}
We thank the anonymous referees for their insightful comments and suggestions that improved the presentation of this paper.
\bibliographystyle{abbrv}
\bibliography{minrk_hard}

\end{document}